\documentclass[runningheads]{llncs}

\usepackage{amssymb,amsmath,amsfonts,mathrsfs}
\usepackage[utf8]{inputenc}
\usepackage[english]{babel}
\usepackage{graphicx}
\usepackage{color}
\usepackage{doi}
\usepackage{algorithmic}
\usepackage{algorithm}

\DeclareMathOperator{\rank}{rank}

\DeclareMathOperator{\poly}{poly}

\DeclareMathOperator{\MAXIS}{MAX-IS}

\DeclareMathOperator{\BUnit}{\mathbf 1}

\DeclareMathOperator{\Coast}{\mathbf{C}}

\DeclareMathOperator{\ZZ}{\mathbb{Z}}
\DeclareMathOperator{\QQ}{\mathbb{Q}}
\DeclareMathOperator{\RR}{\mathbb{R}}

\newcommand*{\intint}[2][1]{#1\!:\!#2}

\begin{document}

\title{An FPTAS for the $\Delta$-modular Multidimensional Knapsack Problem\thanks{The article was prepared under financial support of Russian Science Foundation grant No 21-11-00194.} }
%
%
\author{D.~V.~Gribanov \orcidID{0000-0002-4005-9483}}
%
%
\institute{National Research University Higher School of Economics, 25/12 Bolshaja Pecherskaja Ulitsa, Nizhny Novgorod, 603155, Russian Federation \email{dimitry.gribanov@gmail.com} }
\maketitle              
\begin{abstract}
It is known that there is no EPTAS for the $m$-dimensional knapsack problem unless $W[1] = FPT$. It is true already for the case, when $m = 2$. But, an FPTAS still can exist for some other particular cases of the problem.

In this note, we show that the $m$-dimensional knapsack problem with a $\Delta$-modular constraints matrix admits an FPTAS, whose complexity bound depends on $\Delta$ linearly. 
More precisely, the proposed algorithm arithmetical complexity is $O(n \cdot (1/\varepsilon)^{m+3} \cdot \Delta)$, for $m$ being fixed. Our algorithm is actually a generalisation of the classical FPTAS for the $1$-dimensional case.

Strictly speaking, the considered problem can be solved by an exact polynomial-time algorithm, when $m$ is fixed and $\Delta$ grows as a polynomial on $n$. This fact can be observed combining results of the papers \cite{STEINITZILP,FPT18,ProximityUseSparsity}. We give a slightly more accurate analysis to present an exact algorithm with the complexity bound
$
 O(n \cdot \Delta^{m + 1})
$, for $m$ being fixed. Note that the last bound is non-linear by $\Delta$ with respect to the given FPTAS.

The goal of the paper is only to prove the existence of the described FPTAS, and a more accurate analysis can give better constants in exponents. Moreover, we are not worry to much about memory usage. 
\end{abstract}


\section{Introduction}

\subsection{Basic Definitions And Notations}

Let $A \in \mathbb{Z}^{m \times n}$ be an integer matrix. We denote by $A_{ij}$ the $ij$-th element of the matrix, by $A_{i*}$ its $i$-th row, and by $A_{*j}$ its $j$-th column. The set of integer values from $i$ to $j$, is denoted by $\intint[i]j=\left\{i, i+1, \ldots, j\right\}$. Additionally, for subsets $I \subseteq \{1,\dots,m\}$ and $J \subseteq \{1,\dots,n\}$, the symbols $A_{IJ}$ and $A[I,J]$ denote the sub-matrix of $A$, which is generated by all the rows with indices in $I$ and all the columns with indices in $J$. If $I$ or $J$ are replaced by $*$, then all the rows or columns are selected, respectively. Sometimes, we simply write $A_{I}$ instead of $A_{I*}$ and $A_{J}$ instead of $A_{*J}$, if this does not lead to confusion.

The maximum absolute value of entries in a matrix $A$ is denoted by $\|A\|_{\max} = \max_{i,j} |A_{i\,j}|$. The $l_p$-norm of a vector $x$ is denoted by $\|x\|_p$. 

\begin{definition}
For a matrix $A \in \ZZ^{m \times n}$, by $$
\Delta_k(A) = \max\{|\det A_{IJ}| \colon I \subseteq \intint m,\, J \subseteq \intint n,\, |I| = |J| = k\},
$$ we denote the maximum absolute value of determinants of all the $k \times k$ sub-matrices of $A$.  
Clearly, $\Delta_1(A) = \|A\|_{\max}$. Additionally, let $\Delta(A) = \Delta_{\rank(A)}(A)$. 
\end{definition}


\subsection{Description of Results and Related Work}

Let $A \in \ZZ_+^{m \times n}$, $b \in \ZZ_+^m$, $c \in \ZZ_+^n$ and $u \in \ZZ^{n}_+$. \emph{The bounded $m$-dimensional knapsack problem (shortly $m$-BKP)} can be formulated as follows: 
\begin{gather}
    c^\top x \to \max \notag\\
    \begin{cases}
    A x \leq b \\
    0 \leq x \leq u \\
    x \in \ZZ^n.
    \end{cases}\tag{$m$-BKP}\label{main_prob}
\end{gather}

It is well known that the $m$-BKP is $NP$-hard already for $m = 1$. However, it is also well known that the $1$-BKP admits an FPTAS. The historically first FPTAS for the $1$-BKP was given in the seminal work of O.~Ibarra and C.~Kim \cite{IK}. The results of \cite{IK} were improved in many ways, for example in the works \cite{Chan,ParamKP,ParamWKP,FastUKP,CeJin,KP1,KP2,LAW,CarMKP,MO,Rhee}. But, it was shown in \cite{NoFPTAS} (see \cite[p.~252]{TheMKP} for a simplified proof) that the $2$-BKP does not admit an FPTAS unless $P = NP$. Due to \cite{NoEPTAS}, the $2$-BKP does not admit an EPTAS unless $W[1] = FPT$. However, the $m$-BKP still admits a PTAS. To the best of our knowledge, the state of the art PTAS is given in \cite{BestPTAS}. The complexity bound proposed in \cite{BestPTAS} is $O(n^{\lceil \frac{m}{\varepsilon} \rceil - m})$. The perfect survey is given in the book \cite{TheMKP}.

Within the scope of the article, we are interested in studying $m$-BKP problems with a special restriction on sub-determinants of the constraints matrix $A$. More precisely, we assume that all rank-order minors of $A$ are bounded in an absolute value by $\Delta$. We will call this class of $m$-BKPs as \emph{$\Delta$-modular $m$-BKPs}. The main result of the paper states that the $\Delta$-modular $m$-BKP admits an FPTAS, whose complexity bound depends on $\Delta$ linearly, for any fixed $m$. 

\begin{theorem}\label{main_th}
The $\Delta$-modular \ref{main_prob} admits an FPTAS with the arithmetical complexity bound
$$
O(T_{LP} \cdot (1/\varepsilon)^{m+3} \cdot (2m)^{2m + 6} \cdot \Delta),
$$  where $T_{LP}$ is the linear programming complexity bound.
\end{theorem} Proof of the theorem is given in Section \ref{proof_sec}.

Due to the seminal work of N.~Megiddo \cite{MEG}, the linear program can be solved by a linear-time algorithm if $m$ is fixed. 
\begin{corollary}\label{main_cor}
For fixed $m$ the complexity bound of Theorem \ref{main_th} can be restated as 
$$
O(n \cdot (1/\varepsilon)^{m+3} \cdot \Delta).
$$
\end{corollary}

We need to note that results of the papers \cite{STEINITZILP,FPT18,ProximityUseSparsity} can be combined to develop an exact polynomial-time algorithm for the considered $\Delta$-modular $m$-BKP problem, and even more, for any $\Delta$-modular ILP problem in standard form with a fixed number of constraints $m$. But, the resulting algorithm complexity contains a non-linear dependence on $\Delta$ in contrast with the developed FPTAS. The precise formulation will be given in the following Theorem \ref{Delta_ILP_th} and Corollary \ref{Delta_ILP_cor}. First, we need to make some definitions:
\begin{definition}
Let $A \in \ZZ^{m \times n}$, $b \in \ZZ^m$, $c \in \ZZ^n$, $u \in \ZZ^n_{+}$, $\rank(A) = m$ and $\Delta = \Delta(A)$. \emph{The bounded $\Delta$-modular ILP in standard form} (shortly $m$-BILP) can be formulated as follows:
\begin{gather}
c^\top x \to \max \notag \\    
\begin{cases}
        A x = b \\
        0 \leq x \leq u \\
        x \in \ZZ^{n}.
\end{cases} \tag{$m$-BILP}\label{ILP_standard}
\end{gather}

The main difference between the problems \ref{ILP_standard} and \ref{main_prob} is that the input of the problem \ref{ILP_standard} can contain negative numbers. The inequalities of the problem \ref{main_prob} can be turned to equalities using slack variables.
\end{definition}

\begin{definition}\label{H_def}
Consider the problem \ref{ILP_standard}. Let $z^*$ be an optimal solution of \ref{ILP_standard} and $x^*$ be an optimal vertex-solution of the LP relaxation of \ref{ILP_standard}. \emph{The $l_1$-proximity bound $H$} of the problem \ref{ILP_standard} is defined by the formula
$$
H = \max_{x^*}\min_{z^*} \|x^* - z^*\|_1.
$$

It is proven in \cite{STEINITZILP} that
\begin{equation}\label{Delta1_H_bound}
    H \leq m \cdot (2 m \cdot \Delta_1 + 1)^m,\quad\text{where $\Delta_1 = \Delta_1(A) = \|A\|_{\max}$.}
\end{equation}

It was noted in \cite[formula (4)]{ProximityUseSparsity} that this proximity bound \eqref{Delta1_H_bound} of the paper \cite{STEINITZILP} can be restated to work with the parameter $\Delta(A)$ instead of $\Delta_1(A)$. More precisely, there exists an optimal solution $z^*$ of the \ref{ILP_standard} problem such that \begin{equation*}\label{Delta_H_bound}
    H \leq m\cdot(2m+1)^m \cdot \Delta,\quad\text{where $\Delta = \Delta(A)$.}
\end{equation*}
\end{definition}

\begin{theorem}\label{Delta_ILP_th}
The \ref{ILP_standard} problem can be solved by an algorithm with the following arithmetical complexity:
\begin{equation*}
    n \cdot {O(H + m)}^{m+1} \cdot \log^2(H) \cdot \Delta + T_{LP}.
\end{equation*}

The previous complexity bound can be slightly improved in terms of $H$:
\begin{gather*}
    n \cdot O(\log m)^{m^2} \cdot {(H + m)}^{m}\cdot \Delta + T_{LP}.
\end{gather*}

\end{theorem}
The proof can be found in Section \ref{proof_Delta_ILP_th}.

\begin{remark}\label{hash_table_rm}
The algorithms described in the proof of Theorem \ref{Delta_ILP_th} are using hash tables with linear expected constructions time and constant worst-case lookup time to store information dynamic tables. An example of a such hash table can be found in the book \cite{CORMEN}. So, strictly speaking, algorithms of Theorem \ref{Delta_ILP_th} are randomized.

Randomization can be removed by using any balanced search-tree, for example, $RB$-tree \cite{CORMEN}. It will lead to additional logarithmic term in the complexity bound.
\end{remark}

Applying the proximity bounds \eqref{Delta_H_bound} and \eqref{Delta1_H_bound} to the previous Theorem \ref{Delta_ILP_th}, we can obtain estimates that are independent of $H$. For example, we obtain the following corollary:
\begin{corollary}\label{Delta_ILP_cor}
The problem \ref{ILP_standard} can be solved by an algorithm with the following arithmetical complexity bound:

\begin{gather*}
     n \cdot O(\log m)^{m^2} \cdot O(m)^{m^2+m} \cdot \Delta^{m+1} + T_{LP}\quad\text{ and }\\
     n \cdot \Delta^{m + 1}, \quad\text{for $m$ being fixed}.
\end{gather*}

\end{corollary}

\begin{remark}\label{BKNAP_rm}
Taking $m = 1$ in the previous corollary we obtain the
$
O(n \cdot \Delta^2)
$ complexity bound for the classical bounded knapsack problem, where $\Delta$ is the maximal absolute value of item weights. Our bound is better than the previous state of the art bounds $O(n^2 \cdot \Delta^2)$ and $O(n \cdot \Delta^2 \cdot \log^2 \Delta)$ due to \cite{STEINITZILP}.
\end{remark}

Better complexity bound for searching of an exact solution can be achieved for the unbounded version of the \ref{ILP_standard} problem. More precisely, for this case, the paper \cite{CONVILP} gives the complexity bound  
$$
O(\sqrt{m} \Delta)^{2m} + T_{LP}.
$$ We note that the original complexity bound from the work \cite{CONVILP} is stated with respect to the parameter $\Delta_1(A) = \|A\|_{\infty}$ instead of $\Delta(A)$ (see the next Remark \ref{Delta1_complexity_rm}), but, due to Lemma 1 of \cite{FPT18}, we can assume that $\Delta_1(A) \leq \Delta(A)$. 

\begin{remark}\label{Delta1_complexity_rm}
Another interesting parameter of the considered problems \ref{main_prob} and \ref{ILP_standard} is $\Delta_1(A) = \|A\|_{\max}$. Let us denote $\Delta_1 = \Delta_1(A)$. The first exact quasipolynomial-time algorithm for \ref{ILP_standard} was constructed in the seminal work of C.~H.~Papadimitriou \cite{PAPA}. The result of \cite{PAPA} was recently improved in \cite{STEINITZILP}, where it was shown that the \ref{ILP_standard} can be solved exactly by an algorithm with the arithmetical complexity
\begin{gather}
    n \cdot O(m)^{(m+1)^2} \cdot \Delta_1^{m(m+1)} \cdot \log^2(m \Delta_1) + T_{LP}\quad\text{ and }\notag\\
    n \cdot \Delta_1^{m(m+1)} \cdot \log^2(\Delta_1), \quad\text{ for $m$ being fixed}.\label{Delta1_ILP_bound}
\end{gather}
Due to the results of \cite{CONVILP}, the unbounded version of the problem can be solved by an algorithm with the arithmetical complexity
\begin{equation}\label{Delta1_UILP_bound}
O(\sqrt{m} \Delta_1)^{2m} + T_{LP}.    
\end{equation}

The results of our note can be easily restated to work with the $\Delta_1$ parameter. Using the inequality \eqref{Delta1_H_bound}, the arithmetical complexity bound of Corollary \ref{Delta_ILP_cor} becomes
\begin{gather*}
    n \cdot O(\log m)^{m^2} \cdot O(m)^{m^2 + m} \cdot \Delta_1^{m(m+1)} + T_{LP}\quad\text{ and }\notag\\
    n \cdot \Delta_1^{m(m+1)}, \quad\text{ for $m$ being fixed},
\end{gather*}
which is slightly better, than the bound \eqref{Delta1_ILP_bound} of \cite{STEINITZILP}. 

Additionally, Corollary \ref{Delta_ILP_cor} gives currently best bound $O(n \cdot \Delta^2_1)$ for the classical $1$-dimensional bounded knapsack problem, see Remark \ref{BKNAP_rm}.

The analogue result can be stated for our FPTAS. Definitely, for $\gamma > 0$ and $M = \{ y = A x \colon x \in \RR_+^n,\, \|x\|_1 \leq \gamma\}$ we trivially have $|M \cap \ZZ^m| \leq (\gamma \Delta_1)^m$. Applying the algorithm from Section \ref{proof_sec} to this analogue of Corollary \ref{DP_width_cor}, it gives an algorithm with the arithmetical complexity
\begin{gather*}
    O(T_{LP} \cdot (1/\varepsilon)^{m+3} \cdot m^{2m + 6} \cdot (2\Delta_1)^{m}) \quad\text{ and }\notag\\
    O(n \cdot (1/\varepsilon)^{m+3} \cdot \Delta_1^{m}), \quad \text{ for $m$ being fixed}.\label{Delta1_FPTAS_bound}
\end{gather*}

For sufficiently large $\varepsilon$ the last bounds give a better dependence on $m$ and $\Delta_1$, than bounds \eqref{Delta1_ILP_bound} from \cite{STEINITZILP}. 

\end{remark}

\begin{remark}[Why $\Delta$-modular ILPs could be interesting?]
It is well known that the Maximal Independent Set (shortly $\MAXIS$) problem on a simple graph $G = (V,E)$ can be formulated by the ILP
\begin{gather}
    \BUnit^\top x \to \max \notag\\
    \begin{cases}
            A(G)\, x \leq \BUnit\\
            x \in \{0,1\}^{|V|},
    \end{cases}\tag{$\MAXIS$}\label{IS_ILP}
\end{gather}
\end{remark} where $A(G) \in \{0,1\}^{|E| \times |V|}$ is the edge-vertex incidence matrix of $G$. Due to the seminal work \cite{MinorsGraphs} $$\Delta(A(G)) = 2^{\nu(G)},$$ where $\nu(G)$ is the odd-cycle packing number of $G$. Hence, the existence of a polynomial-time algorithm for $\Delta$-modular ILPs will lead to the existence of a polynomial-time algorithm for the \ref{IS_ILP} problem for graphs with a fixed $\nu(G)$ value. Recently, it was shown in \cite{BIMODULAR_STRONG} that $2$-modular ILPs admit a strongly polynomial-time algorithm, and consequently, the $\MAXIS \in P$ for graphs with one independent odd-cycle. But, existence of a polynomial-time algorithms even for the $3$-modular or $4$-modular ILP problems is an interesting open question, as well as existence of a polynomial-time algorithm for the \ref{IS_ILP} problem on graphs with $\nu(G) = 2$. Finally, due to \cite{AZ}, if $\Delta(\bar A)$ is fixed, where $\bar A = \binom{\BUnit^\top}{A(G)}$ is the extended matrix of the ILP \ref{IS_ILP}, then the problem can be solved by a polynomial time algorithm. The shorter proof could be found in \cite{GRIBM17,GRIBM18}, as well as analogue results for vertex and edge Maximal Dominating Set problems. For recent progress on the \ref{IS_ILP} problem with respect to the $\nu(G)$ parameter see the papers \cite{BOCK14,STABLE_SET_GENUS,KNOWN_ODD_CYCLES}. 

Additionally, we note that, due to \cite{BOCK14}, there are no polynomial-time algorithms for the \ref{IS_ILP} problem on graphs with $\nu(G) = \Omega(\log n)$ unless the ETH (the Exponential Time Hypothesis) is false. Consequently, with the same assumption, there are no algorithms for the $\Delta$-modular ILP problem with the complexity bound $\poly(s) \cdot \Delta^{O(1)}$, where $s$ is an input size. Despite the fact that algorithms with complexities $\poly(s) \cdot \Delta^{f(\Delta)}$ or $s^{f(\Delta)}$ may still exist, it is interesting to consider existence of algorithms with a polynomial dependence on $\Delta$ in their complexities for some partial cases of the $\Delta$-modular ILP problem. It is exactly what we do in the paper while fixing the number of constraints in ILP formulations of the problems \ref{main_prob} and \ref{ILP_standard}.

Due to the Hadamard's inequality, the existence of an ILP algorithm, whose complexity depends on $\Delta$ linearly, can give sufficiently better complexity bounds in terms of $\Delta_1$, than the bounds of Remark \ref{Delta1_complexity_rm}.




\begin{remark}[Some notes about lower bounds for fixed $m$.]
Unfortunately, there are not many results about lower complexity bounds for the problem \ref{ILP_standard} with fixed $m$. But, we can try to adopt some bounds based on the $\Delta_1$ parameter to our case. For example, the existence of an algorithm with the complexity bound 
$$
2^{o(m)} \cdot 2^{o(\log_2 \Delta)} \cdot \poly(s)
$$ will contradict to the ETH. It is a straightforward adaptation of \cite[Theorem 3]{FOMIN}.

The Theorem 13 of \cite{CONVILP} states that for any $\delta> 0$ there is no algorithm with the arithmetical complexity bound
$$
f(m) \cdot (n^{2-\delta} + \Delta_1^{2m - \delta}),
$$ unless there exists a truly sub-quadratic algorithm for the $(min, +)$-convolution. Using Hadamard's inequality, it adopts to
$$
f(m) \cdot (n^{2-\delta} + \Delta^{2 - \delta/m}).
$$

The best known bound in terms of $m$ is given in \cite[Corollary 2]{TightComplexityLB}. More precisely, the existence of an algorithm with the complexity bound 
$$
2^{o(m \log m)} \cdot \Delta_1^{f(m)} \cdot \poly(s)
$$ will contradict to the ETH. But, we does not know how to adopt it for $\Delta$-modular case at the moment.

Unfortunately, all mentioned results are originally constructed for the version of \ref{ILP_standard} with unbounded variables and it is the main reason, why their bounds are probably weak with respect to the dependence on the $\Delta$ parameter.  And it would be very interesting to construct a lower bound of the form 
$$
f(m) \cdot \Delta^{\Omega(m)} \cdot \poly(s) 
$$ for the \ref{ILP_standard} problem. Additionally, at the moment we does not know any FPTAS lower bounds for the \ref{main_prob} problem. These questions are good directions for future research.

\end{remark}

\section{Proof of the Theorem \ref{main_th}}\label{proof_sec}

\subsection{Greedy Algorithm}

The $1/(m+1)$-approximate algorithm for the \ref{main_prob} is presented in \cite{BestPTAS} (see also \cite[p.~252]{TheMKP}) for the case $u = \BUnit$. This algorithm can be easily modified to work with a generic upper bounds vector $u$.

\begin{algorithm}[H]
    \caption{The greedy algorithm}
    \begin{algorithmic}[1]
        \REQUIRE an instance of the \ref{main_prob} problem;
        \ENSURE return $1/(m+1)$-approximate solution of the \ref{main_prob};
        \STATE compute an optimal solution $x^{LP}$ of the LP relaxation of the \ref{main_prob};
        \STATE $y := \lfloor x^{LP} \rfloor$ --- a rounded integer solution;
        \STATE $F := \{ i \colon x^{LP}_i \notin \ZZ \}$ --- variables with fractional values;
        \RETURN $\Coast_{gr} := \max\{c^\top y, \max_{i\in F}\{c_i\} \}$;
    \end{algorithmic}
\end{algorithm}

Since the vector $x^{LP}$ can have at most $m$ fractional coordinates and  
\begin{equation}\label{greedy_ineq}
    c^\top y + \sum_{i \in F} c_i = c^\top \lceil x^{LP} \rceil \geq c^\top x^{LP} \geq \Coast_{opt},
\end{equation}
 we have $\Coast_{gr} \geq \frac{1}{m+1} \Coast_{opt}$.
 
\subsection{Dynamic Programming by Costs}

The dynamic programming by costs is one of the main tools in many FPTASes for the $1$-BKP. Unfortunately, it probably can not be generalized to work with $m$-BKPs for greater values of $m$. However, such generalizations can exist for some partial cases such as the $\Delta$-modular \ref{main_prob}.

Suppose that we want to solve the \ref{main_prob}, and it is additionally known that $\|x\|_1 \leq \gamma$, for any feasible solution $x$ and some $\gamma > 0$. Then, to develop a dynamic program it is natural to consider only integer points $x$ that satisfy to $\|x\|_1 \leq \gamma$. The following simple lemma and corollary help to define such a program. 

\begin{lemma}\label{DP_width_lm}
Let $A \in \ZZ^{m \times n}$ and $B \in \ZZ^{m \times m}$ be the non-degenerate sub-matrix of $A$. Let additionally $\gamma \in \RR_{>0}$, $\Delta = \Delta(A)$, $\delta = |\det B|$ and 
$$
M = \{y = A x \colon x \in \RR^n,\, \|x\|_1 \leq \gamma \},
$$
$$
\text{then}\quad |M \cap \ZZ^m| \leq 2^m \cdot \lceil 1 + \gamma \cdot \frac{\Delta}{\delta} \rceil^m \cdot \Delta.
$$

Points of $M \cap \ZZ^m$ can be enumerated by an algorithm with the arithmetical complexity bound:
$$
O(m^2 \cdot 2^m \cdot D ),
$$ where $D = \Delta \cdot \left( \gamma \cdot \frac{\Delta}{\delta} \right)^m$.
\end{lemma}
\begin{proof}
W.l.o.g. we can assume that first $m$ columns of $A$ form the sub-matrix $B$. Consider a decomposition $A = B \bigl(I \; U\bigr)$, where $\bigl(I \; U\bigr)$ is a block-matrix, $I$ is the $m \times m$ identity matrix and the matrix $U$ is determined uniquely from this equality. Clearly, $\Delta(\bigl(I \; U\bigr)) = \frac{\Delta}{\delta}$, so $\Delta_k(U) \leq \frac{\Delta}{\delta}$ for all $k \in \intint m$. Consider the set 
$$
N = \{ y = \lceil 1 + \gamma \cdot \frac{\Delta}{\delta} \rceil B x \colon x \in (-1,1)^m\}.
$$
Let us show that $M \subseteq N$. Definitely, if $y = A x$ for $\|x\|_1 \leq \gamma$, then $y = B \bigl(I \; U\bigr) x = B t$, for some $t \in [-\gamma,\gamma]^m \cdot \frac{\Delta}{\delta}$. Finally, $\frac{1}{\lceil 1 + \gamma\cdot \frac{\Delta}{\delta} \rceil} t \in (-1,1)^m$ and $y \in N$. 

To estimate the value $|N \cap \ZZ^m|$ we just note that $N$ can be covered by $2^m$ parallelepipeds of the form $\{ y = Q x \colon x \in [0,1)^m\}$, where $Q \in \ZZ^{m \times m}$ and $|\det Q| = \lceil 1 + \gamma\cdot \frac{\Delta}{\delta} \rceil^m \cdot \Delta$. It is well known that the number of integer points in such parallelepipeds is equal to $|\det Q|$, see for example \cite{SEB99} or \cite[Section~16.4]{SCHR98}. Hence, $|M \cap \ZZ^m| \leq |N \cap \ZZ^m| \leq 2^m \cdot \lceil 1 + \gamma\cdot \frac{\Delta}{\delta} \rceil^m \cdot \Delta$. Points inside of the parallelipiped can be enumerated by an algorithm with arithmetical complexity 
$$
O(m \cdot \min\{\log (|\det Q|), m\} \cdot |\det Q|),
$$ see for example \cite{FPT18}. Applying the last formula, we obtain the desired complexity bound to enumerate all integer points inside $N$.

\end{proof}

\begin{corollary}\label{DP_width_cor}
Let $A \in \ZZ^{m \times n}$, $\gamma \in \RR_{>0}$, $\Delta = \Delta(A)$ and 
$$
M = \{y = A x \colon x \in \RR^n,\, \|x\|_1 \leq \gamma \},
$$
$$
\text{then}\quad |M \cap \ZZ^m| \leq 2^m \cdot \lceil 1 + \gamma \rceil^m \cdot \Delta.
$$

Points of $M \cap \ZZ^m$ can be enumerated by an algorithm with the arithmetical complexity bound:
$$
O(\log m)^{m^2} \cdot \Delta \cdot \gamma^m.
$$
\end{corollary}
\begin{proof}
W.l.o.g. we can assume that $\rank(A) = m$. Let us choose $B \in \ZZ^{m \times m}$, such that $|\det B| = \Delta$, then the desired $|M \cap \ZZ^m|$-bound follows from the previous Lemma \ref{DP_width_lm}. Due to \cite{SUBDET_APPROX}, we can compute a matrix $\hat B \in \ZZ^{m \times m}$ such that $\Delta = {O(\log m)}^m \cdot \delta$, where $\delta = |\det \hat B|$, by a polynomial time algorithm. Finally, we take a complexity bound of the previous Lemma \ref{DP_width_lm} with $D = \Delta \cdot \gamma^m \cdot {O(\log m)}^{m^2}$.
\end{proof}




We note that in the current section we need only first parts of these Lemma \ref{DP_width_lm} and Corollary \ref{DP_width_cor} that only estimate number of points nor enumerate them.

Assume that the goal function of the \ref{main_prob} is bounded by a constant $C$. Then, for any $c_0 \in \intint C$ and $k \in \intint n$ we denote by $DP(k,c_0)$ the set of all possible points $y \in \ZZ^m_+$ that satisfy to the system
\begin{equation*}\label{DP_relations}
    \begin{cases}
    c^\top_{\intint k} x = c_0\\
    y = A_{\intint k} x \\
    A_{\intint k} x \leq b \\
    0 \leq x \leq u_{\intint k}\\
    x \in \ZZ^k.
    \end{cases}
\end{equation*}

In particular, the optimal value of the \ref{main_prob} can be computed by the formula 
$$
c^\top x^{opt} = \max\{c_0 \in [1,C] \cap \ZZ \colon DP(n,c_0) \not= \emptyset\}.
$$

The set $DP(k,c_0)$ can be recursively computed using the following algorithm:
\begin{algorithm}[H]
\caption{An algorithm to compute $DP(k,c_0)$}
\begin{algorithmic}[1]
\FORALL{$z \in [0,\gamma] \cap [0,u_k] \cap \ZZ$}
    \FORALL{$y \in DP(k-1,c_0-z c_k)$}
        \IF{$y + A_k z \leq b$}
            \STATE {\bf add} $y + A_k z$ {\bf into} $DP(k,c_0)$
        \ENDIF
    \ENDFOR
\ENDFOR 
\end{algorithmic}
\end{algorithm}

By Corollary \ref{DP_width_cor}, we have $|DP(k,c_0)| \leq 2^m \cdot \lceil 1 + \gamma \rceil^m \cdot \Delta$. Consequently, to compute $DP(k,c_0)$ we need at most $O(m \cdot (2\gamma)^{m+1} \cdot \Delta)$ arithmetic operations. The total complexity bound is given by the following trivial lemma.
\begin{lemma}\label{DP_complexity_lm}
The sets $DP(k,c_0)$ for $c_0 \in \intint C$ and $k \in \intint n$ can be computed by an algorithm with the arithmetical complexity
$$
O(n \cdot C \cdot m \cdot (2\gamma)^{m+1} \cdot \Delta).
$$
\end{lemma}

\subsection{Putting Things Together}

Our algorithm is based on the scheme proposed in the seminal work \cite{IK} of O.~Ibarra and C.~Kim. Our choice of an algorithmic base is justified by the fact that it is relatively easy to generalize the approach of \cite{IK} to the $m$-dimensional case. On the other hand, more sophisticated schemes described in the papers \cite{KP1,KP2,LAW,MO} give constant improvements in the exponent or improvements in the memory usage only.

First of all, let us define two parameters $\alpha, \beta \in \QQ_{>0}$, whose purpose will be explained later. Let $\Coast^{gr}$ be the value of the greedy algorithm applied to the original $\Delta$-modular \ref{main_prob}, $x^{opt}$ be its integer optimal point and $\Coast^{opt} = c^\top x^{opt}$. As it was proposed in \cite{IK}, we split items into heavy and light: $H = \{i \colon c_i > \alpha \, \Coast^{gr} \}$ and $L = \{i \colon c_i \leq \alpha \, \Coast^{gr} \}$.

It can be shown that $\|x_H\|_1 \leq \frac{m+1}{\alpha}$ for any feasible solution $x$ of \ref{main_prob}. Definitely, if $\|x_H\|_1 > \frac{m+1}{\alpha}$, then $\Coast^{opt} = c^\top x^{opt} \geq c^\top x \geq c^\top_H x_H > \alpha  \Coast^{gr} \frac{m+1}{\alpha} = (m+1)  \Coast^{gr} \geq \Coast^{opt}$.

Let $s = \beta \Coast^{gr}$, we put $w = \lfloor \frac{c}{s} \rfloor$. Consider a new $\Delta$-modular $m$-BKP that consists only from heavy items of the original problem with the scaled costs $w$. 
\begin{gather}
w_H^\top x \to \max\notag\\
\begin{cases}
A_H x \leq b\\
0 \leq x \leq u_H\\
x \in \ZZ^{|H|}.
\end{cases}\tag{HProb}\label{H_prob}
\end{gather}

It follows that $\|x\|_1 \leq \frac{m+1}{\alpha}$ for any feasible solution of \eqref{H_prob}. Additionally, we have $w_H^\top x \leq \frac{1}{s} c_H^\top x \leq \frac{m+1}{s} \Coast^{gr} \leq \frac{m+1}{\beta}$, for any $x$ being feasible solution of \eqref{H_prob}. Hence, we can apply Lemma \ref{DP_complexity_lm} to construct the sets $DP_H(k,c_0)$ for $k \in \intint n$ and $c_0 \in \intint \lceil \frac{m+1}{\beta} \rceil$. Due to Lemma \ref{DP_complexity_lm}, the arithmetical complexity of this computation is bounded by 
\begin{equation}\label{applied_DP_complexity}
    O(n \cdot \frac{m (m+1)}{\beta} \cdot (2\frac{m+1}{\alpha})^{m+1} \cdot \Delta ).
\end{equation}

To proceed further, we need to define a new notation $Pr(I, t)$. For a set of indexes $I \subseteq \intint n$ and for a vector $t \in \ZZ_+^m$, we denote by $Pr(I, t)$ the optimal value of the sub-problem, induced by variables with indexes in $I$ and by the right hand side vector $t$. Or by other words, $Pr(I, t)$ is the optimal value of the problem
\begin{gather*}
    c_I^\top x \to \max\\
    \begin{cases}
            A_I x \leq t\\
            0 \leq x \leq u_I\\
            x \in \ZZ^{|I|}.
    \end{cases}
\end{gather*}

After $DP_H(n,c_0)$ being computed we can construct resulting approximate solution, which will be denoted as $\Coast^{apr}$, by the following algorithm.
\begin{algorithm}[H]

\caption{An FPTAS for \ref{main_prob}}

\begin{algorithmic}[1]
\FORALL{$c_0 \in \intint \lceil \frac{m+1}{\beta} \rceil$}
    \FORALL{$y \in DP_H(n,c_0)$}
        \STATE compute an approximate solution $q$ of the problem $Pr(L,b-y)$
        \begin{gather*}
            c_L^\top x \to \max \notag\\
            \begin{cases}
                    A_L x \leq b - y \\
                    0 \leq x \leq u_L \\
                    x \in \ZZ^{|L|}
            \end{cases} \label{L_prob}
        \end{gather*}
        using the greedy algorithm.
        \STATE $\Coast^{apr} := \max\{\Coast^{apr}, s\, c_0 + q\}$.
    \ENDFOR
\ENDFOR 
\end{algorithmic}
\end{algorithm}
Due to Corollary \ref{DP_width_cor}, the arithmetical complexity of the algorithm can be estimated as
\begin{equation}\label{join_complexity}
    O( T_{LP}\cdot \frac{m}{\beta} \cdot (2 \frac{m+1}{\alpha})^m \cdot \Delta ).
\end{equation}

Clearly, $x^{opt} = x_H^{opt} + x_L^{opt}$. We denote $C^{opt}_H = c^\top_H x^{opt}_H$, $C^{opt}_L = c^\top_L x^{opt}_L$ and $c_0^* = w^\top_H x^{opt}_H$. The value of $c_0^*$ will arise in some evaluation of Line 1 of the proposed algorithm. Or by other words, we will have $c_0 = c_0^*$ in some evaluation of Line 1. Let $y*$ be the value of $y \in DP_H(n,c_0^*)$ such that $s\,c^*_0 + Pr(L,b-y^*)$ is maximized and $q^*$ be the approximate value of $Pr(L,b-y^*)$, given by the greedy algorithm in Line 3. Clearly, $\Coast^{apr} \geq s\,c_0^* + q^*$, so our goal is to chose parameters $\alpha,\beta$ in such a way that the inequality $s\,c_0^* + q^* \geq (1-\varepsilon) \Coast^{opt}$ will be satisfied.

Firstly, we estimate the difference $\Coast^{opt}_H - s\, c_0^*$: 
\begin{multline*}
    \Coast^{opt}_H - s\, c_0^* \leq (c^\top_H - s\, w^\top_H) x^{opt}_H \leq s\,\{c^\top_H/s\} x^{opt}_H \leq \\
    \leq s \frac{m+1}{\alpha} = \frac{(m+1) \beta \Coast_{gr}}{\alpha } \leq \frac{(m+1) \beta \Coast_{opt}}{\alpha}
\end{multline*}

To estimate the difference $\Coast^{opt}_L - q^*$ we need to note that $Pr(L,b-y^*) \geq Pr(L, b - A_H x^{opt}_H) = \Coast^{opt}_L$. It follows from optimality of $y^*$ with respect to the developed dynamic program. Next, since $c_i \leq \alpha \Coast^{gr}$ for $i \in L$, due to the inequality \eqref{greedy_ineq}, we have 
\begin{equation*}
    q^* \geq Pr(L,b-y^*) - m\, \alpha \Coast^{gr} \geq Pr(L,b-y^*) - (m+1)\, \alpha \Coast^{opt}.
\end{equation*} 
Finally, we have 
\begin{equation*}
    \Coast^{opt}_L - q^* \leq Pr(L,b-y^*) - q^* \leq (m+1)\, \alpha \Coast^{opt}.
\end{equation*}

Putting all inequalities together, we have 
\begin{multline*}
    \Coast^{opt} - \Coast^{apr} \leq (\Coast^{opt}_H - s\,c_0^*) + (\Coast^{opt}_L - q^*) \leq \\
    \leq (m+1)(\alpha + \frac{\beta}{\alpha}) \Coast^{opt}
\end{multline*} and 
\begin{equation}\label{final_error}
    \Coast^{apr} \geq (1-(m+1)(\alpha + \frac{\beta}{\alpha}))\Coast^{opt}.
\end{equation}

The total arithmetical complexity can be estimated as
\begin{equation}\label{final_complexity}
    O(T_{LP} \cdot \frac{1}{\beta} \cdot (2 m)^{m+3} \cdot \left(\frac{1}{\alpha}\right)^{m+1} \cdot \Delta).
\end{equation}

Finally, after the substitution $\beta = \alpha^2$ and $\alpha = \frac{\varepsilon}{2(m+1)}$ to \eqref{final_error} and \eqref{final_complexity}, we have
$$
\Coast^{apr} \geq (1 - \varepsilon) \Coast^{opt}
$$ and a complexity bound
$$
O(T_{LP} \cdot (1/\varepsilon)^{m+3} \cdot (2m)^{2m + 6} \cdot \Delta)
$$ that finishes the proof.








\section{Proof of Theorem \ref{Delta_ILP_th}}\label{proof_Delta_ILP_th}

Let $x^*$ be an optimal vertex solution of the LP relaxation of the $\Delta$-modular \ref{ILP_standard} problem. After a standard change of coordinates $x \to x - \lfloor x^* \rfloor$ the original \ref{ILP_standard} transforms to an equivalent ILP with different lower and upper bounds on variables and a different right-hand side vector $b$. For the sake of simplicity we assume that lower bounds of the new problem are equal to zero. 

Any optimal vertex solution of the LP problem has at most $m$ non-zero coordinates, so we have the following bound on the $l_1$-norm of an optimal ILP solution $z^* - \lfloor x^* \rfloor$ of the new problem:
\begin{equation*}
    \|z^* - \lfloor x^* \rfloor\|_1 \leq \|x^* - z^*\|_1 + \|x^* - \lfloor x^* \rfloor\|_1 \leq H + m.
\end{equation*}

\subsection{First Complexity Bound}

Consider a weighted digraph $G = (V,E)$, whose vertices are triplets $(k,h,l)$, for $k \in \intint n$, $l \in \intint[0]{(H + m)}$ and $h \in \{A x \colon \|x\|_1 \leq l \} \cap \ZZ^m$. Using Corollary \ref{DP_width_cor}, we bound the number of vertices $|V|$ by $O(n \cdot 2^m \cdot (H + m)^{m+1} \cdot \Delta )$. By definition, any vertex $(k, h, l)$ has an in-degree equal to $\min\{u_k, l\}+1$. More precisely, for any $j \in \intint[0]\min\{u_k, l\}$ there is an arc from $(k-1, h - A_k j, l - j)$ to $(k, h, l)$, this arc is weighted by $c_k j$. Note that vertex $(k-1, h - A_k j, l - j)$ exists only if $j \leq l$. Additionally, we add to $G$ a starting vertex $s$, which is connected with all vertices of the first level $(1,*,*)$, weights of this arcs correspond to solutions of $1$-dimensional sub-problems. Clearly, the number of arcs can be estimated by 
$$
|E| = O(|V| \cdot (H+m)) = O(n \cdot 2^m \cdot (H + m)^{m+2} \cdot \Delta ).
$$

The \ref{ILP_standard} problem is equivalent to searching of the longest path starting from the vertex $s$ and ending at the vertex $(n, b, H + m)$ in $G$. Since the graph $G$ is acyclic, the longest path problem can be solved by an algorithm with the complexity bound $O(|V| + |E|) = O(n \cdot 2^m \cdot (H + m)^{m+2} \cdot \Delta )$.

We note that during the longest path problem solving, the graph $G$ must be evaluated on the fly. In other words, the vertices and arcs of $G$ are not known in advance, and we build them online. To make constant-time access to vertices we can use a hash-table data structure with constant-time insert and search operations (see Remark \ref{hash_table_rm}).


Finally, using the binarization trick, described in the work \cite{STEINITZILP}, we can significantly decrease the number of arcs in $G$. The idea of the trick is that any integer $j \in [0,\min\{u_k, l\}]$ can be uniquely represented using at most $O(\log^2 (\min\{u_k, l\})) = O(\log^2 (H+m))$ bits. More precisely, for any interval $[0,\min\{u_k, l\}]$ there exist at most $O(\log^2 (H+m))$ integers $s(k,i)$ such that any integer $j \in [0,\min\{u_k, l\}]$ can be uniquely represented as 
\begin{gather*}
    j = \sum_i s(k,i) x_i,\quad\text{ where $x_i \in \{0,1\}$, and }\\
    \sum_i s(k,i) x_i \in [0, \min\{u_k,l\}], \quad \text{for any $x_i \in \{0,1\}$}.
\end{gather*}
Using this idea, we replace the part of the graph $G$ connecting vertices of the levels $(k-1,*,*)$ and $(k,*,*)$ by an auxiliary graph, whose vertices correspond to the triplets $(i,h,l)$, where $i \in \{0,1, \dots, O(\log^2 (H+m))\}$, and any triplet $(i,h,l)$ has in-degree two. More precisely, the vertex $(i, h, l)$ is connected with exactly two vertices: $(i-1, h, l)$ and $(i-1, h - s(k,i) A_k, l - s(k,i))$. The resulting graph will have at most $O(\log^2 (H + m) |V|)$ vertices and arcs, where $|V|$ corresponds to the original graph. Total arithmetical complexity can be estimated as 
$$
O(n \cdot 2^{O(m)} \cdot (H + m)^{m+1} \cdot \log^2 H \cdot \Delta ).
$$

\subsection{Second Complexity Bound}

Consider a weighted digraph $G = (V, E)$, whose vertices are pairs $(k,h)$, for $k \in \intint n$ and $h \in M := \{A x \colon \|x\|_1 \leq H + m \} \cap \ZZ^m$. The edges of $G$ have the same structure as in the graph from the previous subsection. More precisely, for any $j \in \intint[0]{u_k}$ we put an arc from $(k-1, h - A_k j)$ to $(k, h)$, if such vertices exist in $V$, the arc is weighted by $c_k j$. 

We compute all vertices of $G$ directly, using Corollary \ref{DP_width_cor}. Arithmetical complexity of this step is bounded by $n \cdot {O(\log m)}^{m^2} \cdot (H + m)^m \cdot \Delta$. Due to Corollary \ref{DP_width_cor}, $|V| = n \cdot |M| = O(n \cdot 2^m \cdot (H + m)^{m} \cdot \Delta)$ and $|E| = O(n \cdot 2^m \cdot (H + m)^{m+1} \cdot \Delta)$, since an in-degree of any vertex in $G$ is bounded by $H + m +1$. 



Let us fix some vertex-level $(k, *)$ of $G$ for some $k \in \intint n$, and consider an auxiliary graph $F_k$, whose vertices are exactly elements $h \in M = \{A x \colon \|x\|_1 \leq H + m\} \cap \ZZ^m$. For two vertices $h_1, h_2$ of $F_k$, we put an arc from $h_1$ to $h_2$ if $h_2 - h_1 = A_k$. Since the graph $F_k$ is acyclic and since "in" and "out" degrees of any vertex in $F_k$ are at most one, the graph $F_k$ is a disjoint union of paths. This decomposition can be computed by an algorithm with complexity $O(|V(F_k)|) = O(|M|) = O(2^m \cdot (H + m)^m \cdot \Delta)$. Let $(h_1, h_2, \dots, h_t)$ be some path of the decomposition, and $longest(k,h)$ be the value of the longest path in $G$ starting at $s$ and ending at $(k,h)$. Clearly, for any $i \in \intint t$, the value of $longest(k,h_i)$ can be computed by the formula
\begin{equation}\label{longest}
    longest(k,h_i) = \max\limits_{j \in \min\{u_k, i-1\}} longest(k-1,h_{i-j}) + c_k j.
\end{equation}

Consider a queue $Q$ with operations: $Enque(Q,x)$ that puts an element $x$ into the tail of $Q$, $Decue(Q)$ that removes an element $x$ from the head of $Q$, $GetMax(Q)$ that returns maximum of elements of $Q$. It is known fact that queue can be implemented such that all given operations will have amortized complexity $O(1)$. Now, we compute $longest(k,h_i)$, for $h_i \in (h_1, h_2, \dots, h_t)$ using the following algorithm:
\begin{algorithm}[H]
\caption{Compute longest path with respect to $(h_1, h_2, \dots, h_t)$}
\begin{algorithmic}[1]
\STATE Create an empty queue $Q$;
\STATE $w := \min\{u_k, t\}$;
\FOR{$j := 0$ {\bf to} $w$}
    \STATE $Enque(Q, longest(k-1,h_{t-j}) + c_k j)$;
\ENDFOR
\FOR{$i := t$ {\bf down to} $1$}
    \STATE $longest(k,h_i) := GetMax(Q) - c_k (t-i)$;
    \STATE $Decue(Q)$;
    \IF{$i \geq w+1$}
        \STATE $Enque(Q, longest(k-1,h_{i-w-1}) + c_k (t - i+1))$;
    \ENDIF
\ENDFOR
\end{algorithmic}
\end{algorithm}

Correctness of the algorithm follows from the formula \eqref{longest}. The algorithm's complexity is $O(t)$. 

Let us estimate the total arithmetical complexity of the whole procedure. It consists from the following parts:
\begin{enumerate}
    \item Enumerating of points in the set $M$. Due to Corollary \ref{DP_width_cor}, the complexity of this part is $O(\log m)^{m^2} \cdot (H+m)^m \cdot \Delta$; 
    \item Constructing the graphs $F_k$ for each $k \in \intint n$. The number of edges and vertices in $F_k$ can be estimated as $O(|M|)$. Hence, due to Corollary \ref{DP_width_cor}, the complexity of this part can be estimated as $O(n \cdot |M|) = O(n \cdot 2^m \cdot (H+m)^m \cdot \Delta)$.
    \item For each $F_k$, compute a path decomposition of $F_k$. For each path in the decomposition, apply an Algorithm 4. The complexity of this part is clearly the same as in the previous step. 
\end{enumerate}

Therefore, the total complexity bound is roughly
$$
n \cdot O(\log m)^{m^2} \cdot (H + m)^m \cdot \Delta.
$$

\section*{Conclusion}
The paper considers the $m$-dimensional bounded knapsack problem \eqref{main_prob} and the bounded ILP in the standard form \eqref{ILP_standard}. For the problem \ref{main_prob} it gives an FPTAS with the arithmetical complexity bound
$$
O(n \cdot (1/\varepsilon)^{m+3} \cdot \Delta),
$$ where $n$ is the number of variables, $m$ is the number of constraints (we assume here that $m$ is fixed) and $\Delta = \Delta(A)$ is the maximal absolute value of rank-order minors of $A$. For details see Theorem \ref{main_th} and Corollary \ref{main_cor}.

For the problem \ref{ILP_standard} it gives an exact algorithm with the complexity bound
$$
O(n \cdot \Delta^{m+1}).
$$ 

Taking $m = 1$ it gives 
$$
O(n \cdot \Delta^2)
$$ arithmetical complexity bound for the classical bounded knapsack problem. For details see Theorem \ref{Delta_ILP_th} and Corollary \ref{Delta_ILP_cor}.


\begin{thebibliography}{50}

\bibitem{AZ} Alekseev,~V.~V., Zakharova,~D.~V.: Independent sets in the graphs with bounded minors of the extended incidence matrix. Journal of Applied and Industrial Mathematics {\bf 5}, 14--18 (2011) \doi{10.1134/S1990478911010029}






\bibitem{BIMODULAR_STRONG} Artmann,~S., Weismantel,~R., Zenklusen,~R. A strongly polynomial algorithm for bimodular integer linear programming. Proceedings of 49th Annual ACM Symposium on Theory of Computing, pp. 1206--1219 (2017) \doi{10.1145/3055399.3055473}











\bibitem{BOCK14} Bock,~A., Faenza,~Y., Moldenhauer,~C., Vargas,~R., Jacinto,~A.  Solving the stable set problem in terms of the odd cycle packing number. Proceedings of 34th Annual Conference on Foundations of Software Technology and Theoretical Computer Science, Leibniz International Proceedings in Informatics (LIPIcs), vol. 29, 187--198 (2014) \doi{10.4230/LIPIcs.FSTTCS.2014.187}





\bibitem{BestPTAS} Caprara,~A., Kellerer,~H., Pferschy,~U., Pisinger,~D. Approximation algorithms for knapsack problems with cardinality constraints. European Journal of Operational Research {\bf 123}, 333--345 (2000) \doi{10.1016/S0377-2217(99)00261-1}


\bibitem{Chan} Chan,~T. Approximation Schemes for $0-1$ Knapsack. In Proceedings of the 1st Symposium on Simplicity in Algorithms (SOSA), pp. 5:1--5:12 (2018) \doi{10.4230/OASIcs.SOSA.2018.5}


\bibitem{STABLE_SET_GENUS} Conforti,~M., Fiorini,~S., Huynh,~T., Joret,~G., Weltge,~S. The stable set problem in graphs with bounded genus and bounded odd cycle packing number. Proceedings of the 2020 ACM-SIAM Symposium on Discrete Algorithms (SODA), pp. 2896--2915 (2020) \doi{10.1137/1.9781611975994.176}

\bibitem{CORMEN} Cormen,~T.~H., Leiserson,~C.~E., Rivest,~R.~L., Stein,~C. Introduction to Algorithms. 3rd. edition, MIT Press (2009)


\bibitem{FOMIN} Fomin,~F.~V., Panolan,~F., Ramanujan,~M.~S., Saurabh,~S. On the Optimality of Pseudo-polynomial Algorithms for Integer Programming. ESA 2018, pp. 31:1--31:13 (2018) \doi{10.4230/LIPIcs.ESA.2018.31}














\bibitem{STEINITZILP} Eisenbrand,~F., Weismantel,~R. Proximity Results and Faster Algorithms for Integer Programming Using the Steinitz Lemma. ACM Transactions on Algorithms {\bf 16}(1) (2019) \doi{10.1145/3340322}








\bibitem{GRIBM17} Gribanov,~D.~V., Malyshev,~D.~S. The computational complexity of three graph problems for instances with bounded minors of constraint matrices. Discrete Applied Mathematics {\bf 227}, 13--20 (2017) \doi{10.1016/j.dam.2017.04.025}

\bibitem{GRIBM18} Gribanov,~D.~V., Malyshev,~D.~S. The computational complexity of dominating set problems for instances with bounded minors of constraint matrices. Discrete Optimization {\bf 29}, 103--110 (2018) \doi{10.1016/j.disopt.2018.03.002}

\bibitem{FPT18} Gribanov,~D.~V., Malyshev,~D.~S., Pardalos,~P.~M., Veselov,~S.~I. FPT-algorithms for some problems related to integer programming. Journal of Combinatorial Optimization {\bf 35}(4), 1128--1146 (2018) \doi{10.1007/s10878-018-0264-z}



\bibitem{MinorsGraphs} Grossman,~J.~V., Kulkarni,~D.~M., Schochetman,~I.~E. On the minors of an incidence matrix and its Smith normal form. Linear Algebra Appl {\bf 218}, 213 -- 224 (1995) \doi{10.1016/0024-3795(93)00173-W}


\bibitem{ParamWKP} Halman,~N., Holzhauser,~M., Krumke,~S. An FPTAS for the knapsack problem with parametric weights. Operations Research Letters {\bf 46}(5), 487--491 (2018) \doi{10.1016/j.orl.2018.07.005}

\bibitem{KNOWN_ODD_CYCLES} Har-Peled,~S., Rahul,~S. Two (Known) Results About Graphs with No Short Odd Cycles (2018) \url{https://arxiv.org/abs/1810.01832}



\bibitem{ParamKP} Holzhauser,~M., Krumke,~S. An FPTAS for the parametric knapsack problem. Information Processing Letters {\bf 126}, 43--47 (2017) \doi{10.1016/j.ipl.2017.06.006}



\bibitem{IK} Ibarra,~O.~H., Kim,~C.~E. Fast approximation algorithms for the knapsack and sum of subset problem. Journal of the ACM {\bf 22}, 463--468 (1975) \doi{10.1287/moor.3.3.197}

\bibitem{FastUKP} Jansen,~K., Kraft,~S. A faster fptas for the unbounded knapsack problem. European Journal of Combinatorics {\bf 68}, 148--174 (2018) \doi{10.1016/j.ejc.2017.07.016}

\bibitem{CONVILP} Jansen,~K., Rohwedder,~L. On Integer Programming, Discrepancy, and Convolution (2018) \url{https://arxiv.org/abs/1803.04744}

\bibitem{CeJin} Jin,~Ce. An Improved FPTAS for $0-1$ Knapsack. 46th International Colloquium on Automata, Languages, and Programming (ICALP 2019), pp. 76:1--76:14 (2019) \doi{10.4230/LIPIcs.ICALP.2019.76}


\bibitem{KP1} Kellerer,~H., Pferschy,~U. A new fully polynomial time approximation scheme for the knapsack problem. Journal of Combinatorial Optimization {\bf 3}, 59--71 (1999) \doi{10.1023/A:1009813105532}

\bibitem{KP2} Kellerer,~H., Pferschy,~U. Improved dynamic programming in connection with an FPTAS for the knapsack problem. Journal of Combinatorial Optimization {\bf 8}, 5--11 (2004) \doi{10.1023/B:JOCO.0000021934.29833.6b}

\bibitem{TheMKP} Kellerer,~H., Pferschy,~U., Pisinger,~D. Knapsack Problems. Springer, Berlin, Heidelberg (2004) \doi{10.1007/978-3-540-24777-7}



\bibitem{TightComplexityLB} Knop,~D., Pilipczuk,~M., Wrochna,~M. Tight complexity lower bounds for integer linear programming with few constraints. ACM Transactions on Computation Theory (TOCT) {\bf 12}(3), 1--19 (2020) \doi{10.4230/LIPIcs.STACS.2019.44}



\bibitem{NoFPTAS} Korte,~B., Schrader,~R. On the existence of fast approximation schemes. Nonlinear Programming {\bf 4}, 415--437 (1981) \doi{10.1016/B978-0-12-468662-5.50020-3}


\bibitem{NoEPTAS} Kulik,~A., Shachnai,~H. There is no EPTAS for two-dimensional knapsack. Information Processing Letters {\bf 110}(16), 707--710 (2010) \doi{10.1016/j.ipl.2010.05.031}

\bibitem{LAW} Lawler,~B.~L. Fast approximation algorithms for knapsack problems. Mathematics of Operations Research {\bf 4}, 339--356 (1979) \doi{10.1287/moor.4.4.339}




\bibitem{ProximityUseSparsity} Lee,~J., Paat,~J., Stallknecht,~I., Xu,~L. Improving proximity bounds using sparsity. Combinatorial Optimization. ISCO 2020. Lecture Notes in Computer Science, vol. 12176 (2020) \doi{10.1007/978-3-030-53262-8\_10}

\bibitem{CarMKP} Li,~W., Lee,~J. A Faster FPTAS for Knapsack Problem With Cardinality Constraint (2020) \url{https://arxiv.org/abs/1902.00919}



\bibitem{MO} Magazine,~M.~J., Oguz,~O. A fully polynomial approximation algorithm for the $0-1$ knapsack problem. European Journal of Operational Research {\bf 8}, 270--273 (1981) \doi{10.1016/0377-2217(81)90175-2}










\bibitem{MEG} Megiddo,~N., Tamir,~A. Linear time algorithms for some separable quadratic programming problems. Operations Research Letters {\bf 13}, 203--211 (1993) \doi{10.1016/0167-6377(93)90041-E}








\bibitem{PAPA} Papadimitriou,~C.H. On the complexity of integer programming. Journal of the Association for Computing Machinery {\bf 28}, 765--768 (1981) \doi{10.1145/322276.322287}


\bibitem{DP_OPT} Pferschy.~U. Dynamic programming revisited: Improving knapsack algorithms. Computing {\bf 63}(4), 419--430 (1999) \doi{10.1007/s006070050042}

\bibitem{Rhee} Rhee,~D. Faster fully polynomial approximation schemes for knapsack problems. Master’s thesis, Massachusetts Institute of Technology (2015)

\bibitem{SCHR98} Schrijver,~A. Theory of linear and integer programming. John Wiley \& Sons (1998)

\bibitem{SEB99} Seb\"o,~A. An introduction to empty lattice simplices. In: Cornu\'ejols G., Burkard R.E., Woeginger G.J. (eds) Integer Programming and Combinatorial Optimization. IPCO 1999. Lecture Notes in Computer Science, vol. 1610,  pp. 400--414 (1999) \doi{10.1007/3-540-48777-8\_30}







\bibitem{SUBDET_APPROX} Marco,~Di~S., Friedrich,~E., Faenza,~Y., Moldenhauer,~C. On largest volume simplices and sub-determinants. SODA '15: Proceedings of the twenty-sixth annual ACM-SIAM symposium on Discrete algorithms, p. 315--323 (2015) \doi{10.5555/2722129.2722152}














\end{thebibliography}
\end{document}